\documentclass[a4paper,USenglish,numberwithinsect]{lipics-v2021}
\usepackage{url}
\PassOptionsToPackage{normalem}{ulem}
\usepackage{ulem}

\makeatletter
\pdfoutput=1

\AtBeginDocument{\providecommand\thmref[1]{\ref{thm:#1}}}
\AtBeginDocument{\providecommand\remref[1]{\ref{rem:#1}}}
\AtBeginDocument{\providecommand\corref[1]{\ref{cor:#1}}}
\AtBeginDocument{\providecommand\lemref[1]{\ref{lem:#1}}}
\AtBeginDocument{\providecommand\propref[1]{\ref{prop:#1}}}

\theoremstyle{plain}
\newtheorem{thm}{\protect\theoremname}
\theoremstyle{definition}
\newtheorem{defn}[thm]{\protect\definitionname}
\theoremstyle{plain}
\newtheorem{cor}[thm]{\protect\corollaryname}
\theoremstyle{remark}
\newtheorem{rem}[thm]{\protect\remarkname}
\theoremstyle{plain}
\newtheorem{lem}[thm]{\protect\lemmaname}
\theoremstyle{plain}
\newtheorem{prop}[thm]{\protect\propositionname}

\title{Guarded Successor: A Novel Temporal Logic}
\author{Ohad Asor}{IDNI AG, \url{https://tau.net}}{ohad@idni.org}{}{}
\authorrunning{Ohad Asor}
\Copyright{CC-BY}

\nolinenumbers 

\hideLIPIcs

\makeatother

\providecommand{\corollaryname}{Corollary}
\providecommand{\definitionname}{Definition}
\providecommand{\lemmaname}{Lemma}
\providecommand{\propositionname}{Proposition}
\providecommand{\remarkname}{Remark}
\providecommand{\theoremname}{Theorem}

\begin{document}
\maketitle
\begin{abstract}
We present GS (Guarded Successor), a novel decidable temporal logic
with several unique distinctive features. Among those, it allows infinitely
many data values that come not only with equality but with a somehow
rich theory too: the first-order theory of atomless Boolean algebras.
The language also distinguishes between inputs and outputs, and has
a decision procedure for determining whether for all inputs exist
outputs, at each point of time. Moreover, and maybe most surprisingly,
the data values can be nothing but sentences in GS itself. We also
present a non-temporal fragment called NSO (Nullary Second Order)
that enjoys merely this last property. \textbf{These results are crucial
necessary ingredients in any meaningful design of safe AI.} Finally,
all those results are obtained from a novel treatment of the first-order
theory of atomless Boolean algebras.
\end{abstract}

\section{Introduction}

Traditional computation is temporal manipulation of bits. Bits, are
the elements of the smallest possible Boolean algebra. The construction
here can be seen as a generalization of this into working over certain
infinite Boolean algebras. Decidability of a specification language
in this model is of course much less trivial. Further, we will show
how this generalization can support some very surprising abilities.

In the following, GS (Guarded Successor) is introduced, an innovative
and decidable temporal logic that offers several distinctive features.
First, it accommodates infinitely many data values, enhanced by a
complex theory: the first-order theory of atomless Boolean algebras.
Second, the language differentiates between input and output variables,
and allows a decision procedure to prove that for all inputs there
exist outputs, at each point in time. Third, and perhaps most surprisingly,
the data values can be sentences in GS itself. The language is closed
under Boolean combinations and allows quantification over both data
values and time points. Its decision procedure is of a uniquely simplistic
and elegant nature, and differs very much from other common decision
procedures. It relies on the ability to enhance certain languages
with recurrence relations (a form of fixed-point operators), in particular
an extension of the first order theory of atomless Boolean algebras.

We also present NSO (Nullary Second Order logic, a name that was selected
during the first incarnations of the idea and perhaps has to be reconsidered),
a non-temporal fragment that maintains the above third property. All
of these findings stem from a new approach to the first-order theory
of atomless Boolean algebras. These results are essential for designing
safe AI systems.

For ease of understanding we will first introduce the non-temporal
NSO logic and afterwards we will introduce GS, which can be seen as
a temporal extension of NSO. All nontrivial proofs not appearing in the main 
text, appear in the appendix.

Ongoing implementation of the languages described in this paper appears at 
the repository\\
\url{https://github.com/idni/tau-lang}.

\subparagraph*{Intellectual Property}

The methods described here are protected from unauthorized use by
IDNI Inc intellectual property rights, including patents. However,
IDNI Inc. grants permission to use the methods for free in the following
specific enumerated non-commercial instances: personal use, educational
use , and academic purposes. The enumerated non-commercial instances
do not include creation of open-source software that is distributed
to others (whether for free or otherwise).

\subparagraph*{Acknowledgements}

I would like to thank Enrico Franconi, Paweł Parys, and Lucca Tiemens
for their review of this material and plenty of useful discussions.

\subsection{NSO}

The goal of NSO is to have a language that can speak about its own
sentences in a consistent and decidable way. Tarski's \emph{Undefinability
of Truth} has shown that this is impossible under a certain broad
setting. The key of NSO is to abstract sentences, so much so, that
they make merely Boolean algebra (BA) elements. In particular, there
is no access to the syntax of the sentences (in contrast to Tarski's
setting which relies on Gödel numbers), and logically equivalent sentences
are identified.

Any classical logic closed under Boolean combinations makes a BA called
the \emph{Lindenbaum-Tarski Algebra (LTA)} of that logic. Recall that
this is only up to logical equivalence. Now observe two important
points: 1. Any such logic that has an infinite signature (whether
constant, relation, or function symbols), makes an atomless BA. 2.
All countable atomless BAs are isomorphic (which is a well known theorem),
and moreover, all atomless BAs are elementarily equivalent, as proved
by Tarski. Clearly all sentences in languages of interest are finite
strings over a finite alphabet, hence countable. \uline{The} countable
atomless BA is therefore the LTA of major logics of interest.

When we say ``the theory of BA interpreted in a fixed BA $\mathcal{B}$''
we mean not only the first order theory of BA interpreted in $\mathcal{B}$
(recall that an interpretation is a mapping taking symbols from the
signature to actual objects in a structure), but we also mean that
its signature is equipped with constants that are interpreted in each
element of $\mathcal{B}$, so each element has a unique constant assigned
to it. We will refer to those constants as the \emph{interpreted constants}.

Fix a language $\mathcal{L}$ that its LTA makes an atomless BA. Let
NSO$\left[\mathcal{L}\right]$ be the first-order theory of BA interpreted
in that LTA, so each sentence in $\mathcal{L}$ is a constant symbol
in NSO$\left[\mathcal{L}\right]$. So far, NSO$\left[\mathcal{L}\right]$
is a language that speaks about $\mathcal{L}$, but still not about
itself. To this end, first we make the LTA of NSO$\left[\mathcal{L}\right]$
be an atomless BA as well (as currently it is only the two-element
BA, as any logic that is interpreted in a fixed structure). This can
be done by adding infinitely many uninterpreted constant symbols (the
\emph{uninterpreted constants}), or any other such trick. Then, the
interpreted constants are extended to include sentences in NSO$\left[\mathcal{L}\right]$
(this is well-founded by introducing \emph{curly brackets} as below).
Since both $\mathcal{L}$ and NSO$\left[\mathcal{L}\right]$ make
an atomless BA, they are elementarily equivalent under the signature
of BA. By that we can make NSO$\left[\mathcal{L}\right]$ speak (including
quantify) over its own sentences. Further, NSO$\left[\mathcal{L}\right]$
is decidable iff $\mathcal{L}$ is decidable.

\subsection{GS}

As an intuitive starting point, any formula with two free variables,
in any logic, can be seen as defining a set of sequences: we say that
a sequence $s$ models $\phi\left(x,y\right)$ iff any two consecutive
elements $s_{i-1},s_{i}$ in the sequence satisfy $\phi\left(s_{i-1},s_{i}\right)$
(we can interpret $\phi$ in a fixed model, or one may appeal to any
suitable notion of satisfiability). We then write $s\models\phi\left(x,y\right)$.
Now consider the class of logics having the following property: fix
a finite set of constant and variable symbols. Then the set of formulas
making use only of those constant and free variable symbols (we allow,
and require, arbitrarily many quantified variables), up to logical
equivalence, is finite. Here, the most relevant such logic is the
theory of atomless BA (whether or not interpreted in a fixed BA, and
when it does, it is equipped with infinitely many interpreted constants
as above, and this is the nontrivial case).

Denote by $\left|s\right|$ the length of $s$. Given $\phi\left(x,y\right)$,
consider the following process: ask whether exists $s$ s.t. $\left|s\right|=2$
and $s\models\phi$, then whether exists $s$ s.t. $\left|s\right|=3$
and $s\models\phi$, and so on. This series of questions may take
the form of a recurrence relation $\phi_{n}\left(x\right):=\exists y.\phi_{n-1}\left(y\right)\wedge\phi\left(x,y\right)$
with base-case $\phi_{2}\left(x\right):=\exists y.\phi\left(x,y\right)$
(though our preferred form will be slightly different). Then $\phi_{n}\left(x\right)$
means ``exists a sequence of length $n$ starting with $x$'', and
then to get a final answer (per each $n$) we of course need to consider
$\exists x.\phi_{n}\left(x\right)$. Due to the finiteness property
above, this series of questions is going to loop (i.e. at one point,
a logically equivalent formula will occur), and even reach a fixed
point due to the monotonic nature of the setting. We obtain a result
of the form: ``if a sequence of length $N$ exists, then a sequence
of any larger length exists''. It is easy to see that this implies
the existence of an infinite sequence as well.

For now we mention only two additional points, which are apparently
unique to this language in the landscape of decidable temporal logics:
\begin{enumerate}
\item Seen as a program specification language, those sequences are actually
\emph{outputs} or \emph{states}, however we'd like to support \emph{inputs}
as well. This means that we'd like to prove that for each input, at
each point of time, exists an output, that does not depend on future
inputs (\emph{time-compatible}). So we can deal with formulas of the
form $\phi\left(x_{n},x_{n-1},y_{n},y_{n-1}\right)$ where $x_{n},x_{n-1}$
are the current and previous inputs, respectively, and similarly for
the outputs $y_{n},y_{n-1}$. Observe the \emph{bounded lookback}
in this formula, and observe that $n$ can be seen as a free variable
of sort $\mathbb{N}$ which is implicitly quantified universally.
The quantifier pattern for the inputs and outputs would look like
$\forall x_{1}\exists y_{1}\forall x_{2}\exists y_{2}\dots$. It is
easy to express it as a recurrence relation similar to the above,
and again use the finiteness property as above.
\item Allowining GS to operate over the LTA of its own sentences, in the
fashion of NSO, so $\phi$ is in the language of atomless BA, gives
us a software specification language where inputs and outputs may
be sentences in this very same language. This allows, for the first
time, support implementations of the form: ``reject a software update
if it doesn't satisfy certain desired properties'' where the currently
running program is written in the same language as the update, as
well as those ``desired properties''. It is therefore a crucial
ingredient in AI safety. Fortunately, even decidability is preserved.
\end{enumerate}

\section{The Theory of Atomless Boolean Algebras}

We assume that the reader is familiar with the definition of atomless
BAs. We will denote the Boolean operations by $\cup,\cap,'$ (disjunction,
conjunction, complementation, respectively) to distinguish them from
logical connectives.

\subsection{Boolean Functions and Equations}

We follow Rudeanu's terminology.
\begin{defn}
A \emph{Boolean Function} (BF) is a Boolean combination of variables
and constants (from some chosen BA). A \emph{Simple Boolean Function}
(SBF) is a BF where all constants appearing in it are either $0$
or $1$.
\end{defn}

Note that this is in contrast to the common definition of BF, which
is typically defined as what we call SBF. Also note that SBFs can
be understood over any BA, but a BF which is not an SBF has to be
understood over a single BA.

In the standard fashion, constant symbols in the theory of BA are
only $0,1$. However we deal with a much richer setting in which the
theory of BA is interpreted in a fixed BA (typically atomless). The
theory is then enhanced with infinitely many constant symbols in its
signature, each uniquely corresponding to each of the BA elements.
We'll refer to those as the \emph{interpreted constants.}

An atomic formula in the language of BA (interpreted in some fixed
BA) is therefore of the form $f\left(X\right)=g$$\left(X\right)$
where $f,g$ are BFs and $X$ is a tuple of variables. Note that this
is equivalent to $f\left(X\right)+g\left(X\right)=0$ where $+$ denotes
the ring sum, alternatively XOR, or symmetric difference (similarly
we'll use juxtaposition for conjunction, mixing set-theoretic and
ring-theoretic notations). So we assume that atomic formulas are of
the form $f\left(X\right)=0$. In the standard formulation of the
theory of BA (that do not involve the interpreted constants), atomic
formulas would have to involve merely SBFs, and not BFs in general.

Any BA induces a partial order defined by $a\leq b\leftrightarrow ab'=0$.
Recall that a BA is atomless iff the following holds: $\forall x.\left(x\neq0\right)\rightarrow\left(\exists y.0<y<x\right)$
For ease of understanding of virtually any BA material, it is useful
to bear in mind Stone's representation theorem for BAs: any BA is
isomorphic to a field of sets. Practically, it means that BA elements
can always be realized as sets, where the Boolean operations coincide
with the usual set operations, and $\leq$ coincides with $\subseteq$.
This justifies our set-theoretic notation. In LTAs, a formula can
be seen as a set of models, which automatically entails identification
under logical equivalence.

Given a system of equations $\bigwedge_{i}\left[f_{i}\left(X\right)=0\right]$
(we shall always assume that all systems contain finitely many equations),
it is easy to see that it is equivalent to a single equation $\left[\bigcup_{i}f_{i}\left(X\right)\right]=0$.
We shall refer to this as \emph{squeezing the positives}. Similar
``squeezing'' is generally not possible for a system of the form
$\bigwedge_{i}\left[f_{i}\left(X\right)\neq0\right]$, at least not
in infinite BAs.
\begin{defn}
A \emph{minterm\index{minterm@\emph{minterm}}} in $n$ variables,
denoted by $X^{A}$, is a product $x_{1}^{a_{1}}x_{2}^{a_{2}}\dots x_{n}^{a_{n}}$
where $A\in\left\{ 0,1\right\} ^{n}$ and $x_{i}^{1}=x_{i};x_{i}^{0}=x_{i}'$.
\end{defn}

If $xy=0$ then we say that $x,y$ are \emph{disjoint}. So $0$ is
disjoint from all elements including itself. Observe that $X^{A}X^{B}=0$
iff $A\neq B$ or $X^{A}=0$ or $X^{B}=0$.
\begin{defn}
A function $\mathcal{B}^{n}\rightarrow\mathcal{B}$ is in \emph{minterm
normal form\index{minterm normal form@\emph{minterm normal form}}}
if it is written as $f\left(X\right)=\bigcup_{A\in\left\{ 0,1\right\} ^{n}}c_{A}X^{A}$
where $c_{A}\in\mathcal{B}$.
\end{defn}

Clearly any function in minterm normal form is a BF. The converse
is also true. cf. \cite{Rud:74book} for the proof of the following theorem:
\begin{thm}
A function $\mathcal{B}^{n}\rightarrow\mathcal{B}$ is a BF iff it
can be written in minterm normal form $f\left(X\right)=\bigcup_{A\in\left\{ 0,1\right\} ^{n}}f\left(A\right)X^{A}$.
\end{thm}

Note that we use $X,A$ as tuples of variables, so the notation $f\left(X\right),f\left(A\right)$
should be clear.
\begin{cor}
\label{cor:A-BF-is}A BF is uniquely determined by its values over
the two-element BA.
\end{cor}

Now we describe Boole's consistency condition. A system of equations
is \emph{consistent} if it has a solution. The following was discovered
by Boole and is actually a case of quantifier elimination:
\begin{thm}
\label{thm:bcc}Let $f:\mathcal{B}^{n}\rightarrow\mathcal{B}$ be
a BF, then $\exists X.f\left(X\right)=0$ iff $\left[\bigcap_{A\in\left\{ 0,1\right\} ^{n}}f\left(A\right)\right]=0$
and $\exists X.\left[f\left(X\right)\neq0\right]$ iff $\left[\bigcup_{A\in\left\{ 0,1\right\} ^{n}}f\left(A\right)\right]\neq0$
\end{thm}

\begin{thm}
\label{thm:sol-interval}Let $f:\mathcal{B}\rightarrow\mathcal{B}$
be a BF s.t. $f\left(0\right)f\left(1\right)=0$, or equivalently,
$\exists x.f\left(x\right)=0$. Then f$\left(x\right)=0$ iff $x=t+f\left(t\right)$
for some $t$, iff $f\left(0\right)\leq x\leq f'\left(1\right)$.
\end{thm}

Our main proof of correctness for our quantifier elimination method
in atomless BA will involve Hall's marriage theorem\index{Hall's marriage theorem}.
We present it here in its set-theoretical version:
\begin{defn}
Let $A_{1},\dots,A_{n}$ be sets, not necessarily distinct. A choice
of elements $a_{1}\in A_{1},\dots,a_{n}\in A_{n}$ such that $a_{i}\neq a_{j}$
for all $i\neq j$ is called a \emph{system of distinct representatives}.
\end{defn}

\begin{thm}
\label{thm:hall}Let $\mathcal{A}=A_{1},\dots,A_{n}$ be a sequence
of sets, not necessarily distinct. Then $\mathcal{A}$ does not have
a system of distinct representatives, iff there exists a subsequence
$\mathcal{B}=B_{1},\dots,B_{m}$ of $\mathcal{A}$ s.t. $\left|\bigcup_{i}B_{i}\right|<m$.
\end{thm}

\begin{rem}
\label{rem:inf-hall}A simple observation which we shall make use
of later on is that a system of distinct representatives exists iff
it exists for the subsequence in which all infinite $A$'s are removed
from it. In other words, infinite sets in a finite collection of sets
don't influence the existence of distinct representatives.

The following theorem is our key step towards quantifier elimination:
\end{rem}

\begin{thm}
\label{thm:main-ineq}Let $X^{A_{1}},\dots,X^{A_{m}}$ be minterms
in $n$ variables, and $b_{1},\dots,b_{m}$ elements in some BA. Then
$\exists X.\bigwedge_{i=1}^{m}\left[X^{A_{i}}\geq b_{i}\right]$ iff
$b_{i}b_{j}=0$ whenever $A_{i}\neq A_{j}$.
\end{thm}

\begin{cor}
The system $\bigwedge_{i\in I}\left[b_{i}X^{A_{i}}\neq0\right]$ has
a solution iff there exists $0<c_{i}\leq b_{i}$ s.t. $c_{i}c_{j}=0$
whenever $A_{i}\neq A_{j}$.
\end{cor}

The corollary reduces the problem of determining consistency of the
above system to a case of \thmref{hall}, once treating each $b_{i}$
as follows: if it can be written as a disjunction of distinct atoms,
then we treat it as a set whose elements are those atoms, and each
$c_{i}$ is a choice of one or more atoms. If $b_{i}\neq0$ and cannot
be written as a union of atoms, then we treat it as an infinite set
and by that it is eliminated from the problem as we have pointed out
in remark \remref{inf-hall}. 
\begin{rem}
\label{rem:th-mnf}Observe that $a\cup b=0\leftrightarrow a=0\wedge b=0$
and recall that each BF can be written as a sum of minterms, or in
DNF (note that writing a BF in DNF is not the same thing as writing
a formula in DNF). This allows an alternative syntax for theories
of BA where atomic formulas are of the form $cX^{A}=0$. We call this
form \emph{minterm normal form}. Note that this is not the same minterm
normal form of BFs, as here it applies to forms of atomic formulas.
\end{rem}

The above results together with the last remark imply with the following
key conclusion:
\begin{cor}
\label{cor:Multivariate-BFs-over}Finitely many multivariate BFs over
an atomless BA have a common nonzero iff none of them is identically
zero.
\end{cor}

\subsection{Quantifier Elimination in Atomless BAs}

Given a formula in the language of BA, we can write it in a way such
that some chosen innermost quantifier is existential. We then convert
everything under that quantifier to DNF. Distributing the existential
over the DNF clauses, and squeezing the positives (as above) in each
clause, we see that if we can eliminate the existential quantifier
from a formula of the form $\exists x.f\left(x\right)=0\wedge\bigwedge_{i}g_{i}\left(x\right)\neq0$,
then we can eliminate all quantifiers. We'll therefore deal only with
such a case.
\begin{lem}
\label{lem:qelim-main}In any BA, the system $f\left(x\right)=0\wedge g\left(x\right)\neq0$
has a solution iff $f\left(0\right)f\left(1\right)=0\wedge\bigwedge_{i}g_{i}\left(x+f\left(x\right)\right)\neq0$
has a solution.
\end{lem}

\begin{thm}
\label{thm:qelim-main}In atomless BA, the system $f\left(x\right)=0\wedge\bigwedge_{i\in I}g_{i}\left(x\right)\neq0$
has a solution iff $f\left(0\right)f\left(1\right)=0\wedge\bigwedge_{i\in I}g_{i}\left(f\left(0\right)\right)\cup g_{i}\left(f'\left(1\right)\right)\neq0$.
\end{thm}

\begin{proof}
Using the last lemma and \thmref{bcc}.
\end{proof}
\begin{prop}
\label{prop:xfx}For any BF $f$ we have $xf\left(x\right)=xf\left(1\right)$
and $x'f\left(x\right)=x'f\left(0\right)$.
\end{prop}

\begin{proof}
Exercise.
\end{proof}
\begin{lem}
\label{lem:egsbe-norm}In any BA, $f\left(x\right)=0\wedge\bigwedge_{i}g_{i}\left(x\right)\neq0$
has a solution iff $f\left(x\right)=0\wedge\bigwedge_{i}xf'\left(1\right)g_{i}\left(1\right)\neq0\vee x'f'\left(0\right)g_{i}\left(0\right)\neq0$
has a solution, iff
\[
\begin{array}{c}
f\left(0\right)f\left(1\right)=0\wedge f\left(x\right)=0\wedge\\
\bigwedge_{i}g_{i}\left(0\right)g_{i}\left(1\right)\neq0\vee xf'\left(1\right)g_{i}\left(1\right)\neq0\vee x'f'\left(0\right)g_{i}\left(0\right)\neq0
\end{array}
\]
has a solution.
\end{lem}

\begin{cor}
In atomless BA, the system $f\left(x\right)=0\wedge\bigwedge_{i}g_{i}\left(x\right)\neq0$
has a solution iff $f\left(0\right)f\left(1\right)=0\wedge\bigwedge_{i}f'\left(1\right)g_{i}\left(1\right)\cup f'\left(0\right)g_{i}\left(0\right)\neq0$.
\end{cor}

\begin{proof}
Applying corollary \corref{Multivariate-BFs-over} to \lemref{egsbe-norm}.
\end{proof}
By that we conclude the final form of our quantifier elimination algorithm
over atomless BAs. Basically we generalized proposition 5.5 in \cite{10.5555/500818}
from SBF to BF over atomless BA (or over any BA whenever the cardinalities
of the constants appearing in the equations are large enough), using
very different and much more direct methods, and we indicated how
and why it fails in non-atomless BA. More treatment for the non-atomless
case can be found at \cite{taba:2024}.

\subsection{Number of Formulas}

We presented minterm normal form of formulas in remark \remref{th-mnf}.
This normal form puts a bound on the number of quantifier-free logically
equivalent formulas with $n$ free variables and $k$ constants. The
accounting is as follows: the formula is itself an SBF of atomic formulas,
and there are $2^{2^{N}}$ different SBFs in $N$ variables. In our
case $N$ is the number of possible minterms which is readily $k2^{n}$.
We therefore end up with a triple exponential $2^{2^{k2^{n}}}$ upper
bound. This extends to quantified formulas as well due to quantifier
elimination.

\subsection{\label{sec:Recurrence-Relations}Recurrence Relations\index{Recurrence Relations@\textbf{Recurrence Relations}}}

We propose the notion of \index{weakly omega-categorical@weakly $\omega$-categorical}\emph{weakly
$\omega$-categorical theories}. Recall that an $\omega$-categorical
theory is a first order theory in which all of its countable models
are isomorphic. The Ryll-Nardzewski theorem says that this definition
is equivalent to another definition: that up to logical equivalence,
there are only finitely many formulas with free variables taken from
a fixed finite set. This gives rise to defining weakly $\omega$-categorical
theories: those are theories for which the number of formulas using
a fixed finite set of free variables and where the constants appearing
in them are taken from a fixed finite subset of all constants in the
language, up to logical equivalence, is finite. For the sake of this
section, it does not matter whether or not the theory is partially
interpreted in a fixed structure. The concept of \emph{partial interpretation}
is as follows: suppose we interpret a logic in a fixed structure.
Now we enhance this logic with additional uninterpreted constants.
Then a partial interpretation would be assigning domain elements to
those newly added constant symbols.

It is easy to see, in light of the previous section and the quantifier
elimination results, that the theory of atomless BA and of fixed finite
BA, are both weakly $\omega$-categorical. In what follows we shall
deal only with those BA theories. However many of the constructions
in this section and across this paper can be carried out into any
weakly $\omega$-categorical theory. 

We are now ready to define formulas in the language of BA (or any
weakly $\omega$-categorical theory) enhanced with recurrence relations.
Any such formula takes the form $\phi_{n}\left(X\right)=\Phi\left(X,\phi_{n-1}\left(X\right),\phi_{n-2}\left(X\right),\dots\right)$
together with suitable base conditions $\phi_{1}\left(X\right)=\dots,\phi_{2}\left(X\right)=\dots,\dots$.
Clearly those recurrence relations do not always have a fixed point,
but weakly $\omega$-categoricity guarantees a partial fixed point.
It is easy to pin down all cases in which a given formula has a fixed
point (in case the theory we start with is decidable), as well as
apply basic remedies for the case that only a partial fixed point
exists, but those are omitted here for the sake of brevity.

We have shown that any formula in the theory of atomless BA enhanced
with recurrence relations can be written in an equivalent form without
recurrence relations.

\section{NSO: Nullary Second Order Logic}

Building on what we said in the introduction, we shall not merely
present a language, but a language-extension mechanism, altough this
extension is not in the standard sense, namely formulas in the base
language are not, as for themselves, formulas in the extended language,
but constant symbols in it, as we shall see. This extension preserves
decidability, let alone consistency. We further consider extending
many languages at once, and it is indeed yet another feature of our
construction to allow languages to co-exist in one unified language,
albeit, of course, the interaction between those languages is very
limited. Referring to many BAs at once is easily done by considering
the many-sorted theory of BA, alternatively the product algebra.

Fix arbitrary languages (the \emph{base logics}) in which their formulas
(or sentences), up to logical equivalence, make a BA. Then we can
consider the many-sorted BA theory interpreted in those BAs. Constants
in that languge are formulas in the base logics. Quantification takes
the same semantics of quantification over arbitrary BA elements. If
the base logics make an atomless BA, then the extended language has
decidable satisfiability iff the base logics have. Otherwise decidable
model counting is required, or more precisely, when seen as a BA,
to tell whether an element is a disjunction of at least $n$ distinct
atoms.

Denote the extended language by $NSO\left[\mathcal{L}_{1},\dots,\mathcal{L}_{n}\right]$.
We show that $NSO\left[\mathcal{L}_{1},\dots,\mathcal{L}_{n}\right]$
can have itself as a base-logic. So far, each NSO formula is either
true or false, because it is interpreted in a fixed model (being the
BA which is the LTA of the base logic), and therfore makes a small
BA (only two elements). To obtain an atomless BA from formulas in
$NSO\left[\mathcal{L}_{1},\dots,\mathcal{L}_{n}\right]$ we can simply
enhance it with infinitely many uninterpreted constant symbols. After
doing so, we now let interepreted constants to be formulas in $NSO\left[\mathcal{L}_{1},\dots,\mathcal{L}_{n}\right]$
appearing inside curly brackets (in order to avoid syntactic ambiguity),
and handling of quantifiers for the sake of a decision procedure can
be done by means of the atomless BA quantifier elimination algorithm.
The basic syntax of $NSO\left[\mathcal{L}_{1},\dots,\mathcal{L}_{n}\right]$
is therefore
\[
\phi:=\exists var:sort.\phi|\phi\wedge\phi|\neg\phi|bf=0
\]
\[
sort:=\mathcal{L}_{1}|\dots|\mathcal{L}_{n}|NSO\left[\mathcal{L}_{1},\dots,L_{n}\right]
\]
\[
bf:=var|\left\{ \phi^{sort}\right\} |const|0|1|bf\cap bf|bf'
\]
where $\phi^{\mathcal{L}}$ means any formula in the language $\mathcal{L}$.
Clearly, each $bf$ may only contain variables and constants from
the same sort. $const$ refers to an uninterpreted constant. The deep-most
level of formulas in {[}nested{]} curly brackets will be either a
formula in $\mathcal{L}_{1},\dots,\mathcal{L}_{n}$ or a formula in
the language of BA in which the only constants appearing in it are
$0,1$. It is then interpreted as a formula over arbitrary atomless
BA since they're all elementarily equivalent. It is easy to see that
going inductively over the depth of curly brackets, gives a decision
procedure as well as semantics to this language.

\section{GSSOTC: A Temporal Logic}

We devise a new, decidable, family of temporal logics over infinite
data values, where those values come with theories much richer than
merely equality, in particular with the theory of atomless Boolean
Algebras (as well as fixed finite ones though such a case does not
amount to a significant novelty). Further, this language enjoys the
distinctive ability to verify statements of the form ``at each point
of time, for all inputs exist a well-defined output/state, possibly
depending on the previous output/state''. It also presents a new
kind of decision procedure, unrelated to automata, tableaux, or to
any other decision method known to the author.

To describe the language in simple intuitive terms: fix an atomless
BA and consider the theory of BA interpreted in this structure (with
interpreted constants as above so the LTA of this logic is the countable
atomless BA). Consider formulas with free variables $x_{n-k},\dots,x_{n},y_{n-k},\dots,y_{n}$
where the $x$'s are understood as inputs and the $y$'s are understood
as outputs, and $n$ is any time point (so it can be seen as a free
variable of sort $\mathbb{N}$). So it describes connection between
current and previous inputs and outputs at each point of time. This
is basically almost the full language.

This technique works for any weakly $\omega$-categorical language,
as long as it supports conjunction and quantification. However in
the atomless BA case we get the unique property of a language that
can speak of its own sentences, in the spirit of NSO.

\subsection{Time-Compatible Structures\index{Time-Compatible (TC) Structure@\textbf{Time-Compatible (TC) Structure}}}

A sequence of elements from some domain $\mathcal{D}$ can be seen
as a function $\mathbb{N}\rightarrow\mathcal{D}$. A function between
sequences is therefore of type $\left(\mathbb{N}\rightarrow\mathcal{D}\right)\rightarrow\left(\mathbb{N}\rightarrow\mathcal{D}\right)$.
As customary in many texts, $\left[k\right]$ will denote the set
$\left\{ 1,\dots,k\right\} $.
\begin{defn}
A function $f:\left(\mathbb{N}\rightarrow\mathcal{D}\right)\rightarrow\left(\mathbb{N}\rightarrow\mathcal{D}\right)$
between sequences is \emph{prefix-preserving} (alternatively \emph{time-compatible},
TC) if for all sequences $p,s$, if $p$ is a \emph{strict} prefix
of $s$, then $f\left(p\right)$ is a \emph{strict} prefix $f\left(s\right)$.
We extend this notion also for $f:\left(\left[n\right]\rightarrow\mathcal{D}\right)\rightarrow\left(\left[n\right]\rightarrow\mathcal{D}\right)$.
\end{defn}

\begin{defn}
A \emph{\index{TC (time-compatible) Structure@\emph{TC (time-compatible) Structure}}Time-Compatible
(TC) Structure of length $N\in\mathbb{N}\cup\left\{ \infty\right\} $}
is a domain $\mathcal{D}$ with prefix-preserving functions $\mathcal{D}^{N}\rightarrow\mathcal{D}^{N}$. 
\end{defn}

It should be clear that any computer program is a TC structure: at
each point of time it takes an input and outputs an output, while
the output may depend only on past and present inputs and outputs,
not future ones. This is why we refer to prefix-preservation as TC.
\begin{rem}
\label{rem:states}Due to the ``lookback'' ability, namely the dependence
on previous inputs and outputs, we don't need to refer to the concept
of state, as it is subsumed by the concept of output.
\end{rem}

\begin{rem}
In what follows we will deal only with infinite-time TC structures
(so $N=\infty$ in the above definition) unless stated otherwise.
\end{rem}

\begin{rem}
\label{rem:tuples}We will eventually be interested with functions
from tuples of sequences to tuples of sequences (all tuples of fixed
finite size, but the input tuple may be of different size than of
the output tuple). All definitions and results should apply mutatis-mutandis.
\end{rem}

\begin{rem}
The setting can easily be extended to trees rather sequences. It is
done by allowing more than one successor relation, and the same methods
apply.
\end{rem}

\begin{defn}
A TC function has \index{bounded lookback, BL}\emph{bounded lookback}
(BL\index{BL, bounded lookback}) of length $k\in\mathbb{N}$ (or
simply BL$\left[k\right]$) if exists $m\geq k$ (the \emph{recurrence
point}), s.t. for each $n>m$, the output sequence at point $n$ depends
only on the input and output sequences at points $n-1,\dots,n-k$,
as well as the input at point $n$.
\end{defn}

\begin{cor}
\label{cor:bl-type}If $f$ is BL$\left[k\right]$ then it can be
expressed as a pair of functions, one of type $\mathcal{D}^{2k+1}\rightarrow\mathcal{D}$
and another of type $\mathcal{D}^{m}\rightarrow\mathcal{D}^{m}$ which
is required to be TC.
\end{cor}

\begin{proof}
By definition of BL functions, we can write $f$ as a recurrence relation
\[
\left[f\left(x\right)\right]_{n}=g\left(x_{n},x_{n-1},\dots,x_{n-k},\left[f\left(x\right)\right]_{n-1},\dots,\left[f\left(x\right)\right]_{n-k}\right)
\]
(where $x$ is the input sequence) with initial conditions of the
form $\left[f\left(x\right)\right]_{i}=\dots$ for $1\leq i\leq k$.
This $g$ is of type $\mathcal{D}^{2k+1}\rightarrow\mathcal{D}$ and
together with the initial conditions (which specify the behavior up
until the recurrence point), fully encodes $f$.
\end{proof}
\begin{cor}
Given a pair of functions, one of type $\mathcal{D}^{2k+1}\rightarrow\mathcal{D}$,
and another, which is TC, of type $\mathcal{D}^{m}\rightarrow\mathcal{D}^{m}$,
we can uniquely assign to it a function of BL$\left[k\right]$.
\end{cor}

\subsection{Bounded Lookback and Recurrence Relations}
\begin{cor}
\label{cor:formula-to-blf}Any formula (in virtually any logic) with
$2k+2$ free variables defines a {[}possibly empty{]} set of BL$\left[k\right]$
functions.
\end{cor}

Note that in the infinitary expression obtained in the proof, quantifiers
can be pushed inside. This is a property of being TC, and this ability
is one cruicial point in the upcoming construction. Also note that
skolemization of this expression will yield something similar to the
type in \corref{bl-type}.
\begin{rem}
The initial conditions are not expressed in the latter corollary.
But the corollary still holds. It defines a set of functions that
include functions per each possible initial condtions. This is not
an inherent limitation. We used this form only for simplicity at this
stage.
\end{rem}

Fix a lookback parameter $k\geq0$. $X_{j}$ will denote a tuple of
variables of lookback $k$, so it's a tuple of $k+1$ variables of
the form $x_{j-k},x_{j-k+1},\dots,x_{j}$. We assume that the first
time coordinate is $0$.
\begin{defn}
Given formula $\phi$ (in virtually any logic) with $2k+2$ free variables
$x_{n-k},\dots,x_{n},y_{n-k},\dots,y_{n}$, define a recurrence relation
$\phi_{n}$ by $\phi_{n+1}\left(X_{k},Y_{k}\right):=\phi\left(X_{k},Y_{k}\right)\wedge\forall x_{k+1}\exists y_{k+1}.\phi_{n}\left(X_{k+1},Y_{k+1}\right)$
with base case $\phi_{1}:=\phi\left(X_{k},Y_{k}\right)$.
\end{defn}

\begin{rem}
\label{rem:rec-sat}Observe that $\phi_{n}\left(X_{k},Y_{k}\right)$
actually says that exists a BL$\left[k\right]$ function between sequences
of length $n+k$, where the $k$ initial positions in the sequences
are left as free variables.
\end{rem}

Note that $\phi_{n}$ has a form of monotonicity wrt $n$: if exists
a TC function between sequences of lentgh $n+1$, and the function
satisfies $\phi$, then clearly exist such a function for sequences
of length $n$.

Clearly, if $\forall x_{0}\exists y_{0}\dots\forall x_{k}\exists y_{k}.\phi_{n}\left(X_{k},Y_{k}\right)$
for all $n$, then $\phi$ defines a nonempty set of functions in
the spirit of corollary \corref{formula-to-blf}. The crux of our
construction is the observation that if the underlying logic is weakly
$\omega$-categorical, then there are only finitely many $\phi_{n}$'s
up to logical equivalence, hence decidability and decision procedure
are immediate.

\subsection{Guarded Successor\index{Guarded Successor@\textbf{Guarded Successor}}}

Observe that a formula of the form $\phi\left(X_{k},Y_{k}\right)$
can be given a direct BL$\left[k\right]$ semantics also by adding
a sort of natural numbers with the successor relation $s$, and function
symbols $f:\mathbb{N}\rightarrow\mathcal{D}$ and $F:\left(\mathbb{N}\rightarrow\mathcal{D}\right)\rightarrow\left(\mathbb{N}\rightarrow\mathcal{D}\right)$,
where $F$ is required to be prefix-preserving, and writing $\phi$
as
\[
\forall t_{0},\dots t_{k}.\left[\bigwedge_{i=0}^{k-1}s\left(t_{i},t_{i+1}\right)\right]\rightarrow\phi\left(f\left(t_{0}\right),\dots,f\left(t_{k}\right),F\left(f\right)\left(t_{0}\right),\dots,F\left(f\right)\left(t_{k}\right)\right)
\]

\begin{defn}
Fix a logic ${\cal L}$ and let $\mathcal{D}$ be the sort it operates
over. First extend it with function symbols $f_{i}:\mathbb{N}\rightarrow\mathcal{D}$
and $F_{j}:\left(\mathbb{N}\rightarrow\mathcal{D}\right)\rightarrow\left(\mathbb{N}\rightarrow\mathcal{D}\right)$,
where $F$ is required to be prefix-preserving. If $\psi$ is any
formula in this extended language, then
\end{defn}

\[
\phi:=\psi|\phi\wedge\phi|\neg\phi|\forall t_{1},\dots,t_{m}.\left[\bigwedge_{\left(i,j\right)\in I}s\left(t_{i},t_{j}\right)\right]\rightarrow\phi
\]
defines a second extension to the language which we shall refer to
as the \emph{guarded successor extention of ${\cal L}$}. The sublanguage
of the form
\[
\phi:=\psi|\phi\wedge\phi|\neg\phi|\forall t_{1},\dots,t_{m}.\left[\bigwedge_{\left(i,j\right)\in I}s\left(t_{i},t_{j}\right)\right]\rightarrow\psi
\]
will be called the \emph{collapsed} \emph{fragment}. Its sublanguage
of the form
\[
\phi:=\bigvee_{k}\left(\forall t_{1},\dots,t_{m}.\left[\bigwedge_{\left(i,j\right)\in I_{k}}s\left(t_{i},t_{j}\right)\right]\rightarrow\psi_{k}^{1}\right)\wedge\left(\exists t_{1},\dots,t_{m}.\left[\bigwedge_{\left(i,j\right)\in J_{k}}s\left(t_{i},t_{j}\right)\right]\wedge\psi_{k}^{2}\right)
\]
will be called the \emph{normalized fragment}. In all cases, the guard
$\bigwedge_{\left(i,j\right)\in I}s\left(t_{i},t_{j}\right)$ is required
to uniquely determines the relative position between each $t_{i},t_{j}$,
and $\psi,\psi_{k}^{1},\psi_{k}^{2}$ involve $t_{1},\dots,t_{m}$
only through application of $f,F$ (or several such functions), while
$f,F$ may also be applied to constants from $\mathbb{N}$.
\begin{rem}
Applying $f,F$ to constants from $\mathbb{N}$ corresponds to the
above initial conditions.
\end{rem}

\begin{thm}
\label{thm:gs-norm}Any formula in a guarded successor extension can
be written as an equisatisfiable formula in the normalized fragment.
\end{thm}

\begin{rem}
Note that here we had to use the assumption that we are dealing with
infinite-time structures, namely $N=\infty$. In the finite-time case
we will also need the end-of-sequence predicate $\sharp$, resulting
with a slightly more complicated quantifier collapse. We omit this
simple derivation here for sake of brevity.
\end{rem}

\begin{cor}
\label{cor:ex-only}Any formula in a guarded successor extension without
temporal existential quantifiers can be written in a free-variable
BL$\left[k\right]$ form $\phi\left(X_{k},Y_{k}\right)$.
\end{cor}

We of course bear in mind that if some language is decidable and is
weakly $\omega$-categorical, then its extension with recurrence relations
is also decidable. Together with a method to handle the existential
part as described in the next section, we'll conclude that:
\begin{cor}
Satisfiability of a formula in a guarded successor extension is decidable
if this fragment is obtained from a decidable language $\mathcal{L}$
which is weakly $\omega$-categorical, enhanced with the sort $\mathbb{N}$,
guarded successors, $\mathbb{N}\rightarrow\mathcal{D}$ function symbols,
and BL$\left[k\right]$ function symbols.
\end{cor}

We refer to this extended language as GSSOTC$\left[\mathcal{L}\right]$,
where GSSOTC\index{GSSOTC} stands for Guarded-Successor Second-Order
Time-Compatible. The second-order part is due to the following: given
two sequences $f,g:\mathbb{N}\rightarrow\mathcal{D}$, we can declare
a non-standard quantifer alternation $\forall f\exists g$, which
would translate into $\exists F\forall f$ (so far just standard higher-order
skolemization), where $F$ is a TC function between sequences. Those
function quantifiers are eliminated when converting the formula to
the free-variable form, which is then converted to function-free recurrence-relation
form.

Some easy extensions of this language were described above, we reiterate
them and add more: the end-of-string predicate $\sharp$, having multiple
successor relations and by that considering trees rather sequences,
having constant positions, so instead of e.g. $\phi\left(x_{n},x_{n-1},y_{n}\right)$,
we have e.g. $\phi\left(x_{1},x_{2},x_{n},x_{n-1},y_{n}\right)$,~having
explicit second-order quantifiers that are eliminated by reduction
to recurrence relations, and finally, having richer quantifier alternation,
e.g. for all keyboard input at time $n$, exists a memory state at
time $n$, s.t. for all network input at time $n$, and so on, resulting
in quantification of the form $\forall x_{1}\exists y_{1}\forall z_{1}\forall x_{2}\exists y_{2}\forall z_{2}\forall x_{3}\exists y_{3}\forall z_{3}\dots$.

\subsection{Decision Methods and Execution}

In the spirit of remark \remref{tuples}, we shall have several input
and output sequences, each referred to as a \emph{stream}.
\begin{thm}
\label{thm:gs-main}Given $\phi\left(X_{j}^{i},Y_{j}^{i}\right)$
where $X$ are inputs and $Y$ are outputs, and $i$ denoting the
stream number, define the recurrence relation
\[
\phi_{0}\left(X_{k}^{i},Y_{k}^{i}\right):=\phi\left(X_{k}^{i},Y_{k}^{i}\right)
\]
\[
\phi_{n}\left(X_{k}^{i},Y_{k}^{i}\right):=\phi\left(X_{k}^{i},Y_{k}^{i}\right)\wedge\forall x_{k+1}\exists y_{k+1}.\phi_{n-1}\left(X_{k+1}^{i},Y_{k+1}^{i}\right)
\]
so $\phi_{n}$ means that exists a model with time points $0,\dots,n+k$
starting with $X_{k}^{i},Y_{k}^{i}$. Then the reccurence relation
is monotonic, namely $\forall n\forall X_{k}^{i}Y_{k}^{i}.\phi_{n+1}\left(X_{k}^{i},Y_{k}^{i}\right)\rightarrow\phi_{n}\left(X_{k}^{i},Y_{k}^{i}\right)$
and therefore has a fixed point. Denote it by $\phi_{\infty}\left(X_{k}^{i},Y_{k}^{i}\right)$.
Given a model of $\phi$ with $m$ time points, and given each input
$X^{i}$ at point $m+1$, then an output $Y^{i}$ will have an unbounded
continuation satisfying $\phi$ iff $\phi_{\infty}\left(X_{m+1}^{i},Y_{m+1}^{i}\right)$.
\end{thm}

\begin{rem}
The above formulation suggests that $\phi_{\infty}$ is a normal form
of $\phi$ when understood as defining TC models.
\end{rem}

\begin{rem}
A TC structure is a model of $\phi$ iff any subsequence satisfies
$\phi_{\infty}$ when understood as a formula in the language of BA.
\end{rem}

\begin{rem}
Given inputs at each point of time, satsifying outputs can be computed
by substituting the known variables into $\phi_{\infty}$, and solving
for the missing outputs. This is an execution method for software
specification in this language. Software specification in this language
is therefore directly executable as-is, using an oracle to determine
satisfying assignments to formulas in the language of atomless BA.
Finding satisfying assignments to a formula in the language of atomless
BA is a topic by its own, and is omitted here for sake of brevity.
\end{rem}

\begin{cor}
Given two formula $\phi\left(X_{j}^{i},Y_{j}^{i}\right),\psi\left(X_{j}^{i},Y_{j}^{i}\right)$,
then the set of TC models of $\phi$ is a subset of those of $\psi$,
iff $\forall x_{0}y_{0}\dots x_{k}y_{k}.\phi_{\infty}\left(X_{k}^{i},Y_{k}^{i}\right)\rightarrow\psi_{\infty}\left(X_{k}^{i},Y_{k}^{i}\right)$.
\end{cor}

This gives us an algorithm to decide whether $\phi\psi'=0$ where
$\phi,\psi$ are seen as sets of TC models.
\begin{rem}
Combined with \thmref{gs-norm} and its proof, this corollary gives
us a decision procedure for the full language GS. Each DNF clause
will have a single universal and a single existential (which is a
negated universal), so deciding emptiness for each clause comes down
to the last corollary.
\end{rem}

\begin{rem}
Since $\phi_{\infty}$ refers only to the universal parts, while the
existential parts may of course restrict the models, therefore we
should, at execution time, check at each point of time whether we
can satisfy the existential parts. If so, we satisfy them indeed,
just once. If the formula is satisfiable then such point in time is
guaranteed to exist. If there are multiple existential parts in a
DNF clause, then for execution, we have to squeeze them into one using
the flags as in the proof of \thmref{gs-norm}, since those existential
parts may depend on each other.
\end{rem}

\begin{rem}
When $\phi\left(X_{n},Y_{n}\right)$ is understood as a GS formula,
and $\phi$ is in the language of atomless BA intepreted in this very
BA of GS formulas (possibly with more algebras as the consrtuction
is closed under products), then NSO is a sublanguage of this language.
That'd be a software specification language where inputs and outputs
are nothing but sentences in this very language. This way we can support
the software update mechanism described in the introduction as a crucial
component for safe AI. Another way to look at it: a robot is programmed
in a language ${\cal L}$ and accepts commands form the user in the
very same language ${\cal L}$. Now its internal program has to ask
whether the command is consistent with, say, safety conditions. It
couldn't do so unless ${\cal L}$ is a temporal logic with inputs
in ${\cal L}$ equipped with the theory of BA.
\end{rem}

\subsection{Complexity}

Quantifier elimination in theories of BA where constants are either
0,1 were studied by Tarski by introducing his so-called invariants.
Kozen \cite{DBLP:journals/tcs/Kozen80} extended this notion of invariants
and by that derived a complexity\index{complexity} characterization
for the decision problem. For infinite BAs, it is complete for $\bigcup_{c}$STA$\left(*,c^{n},n\right)$.
Roughly, this means anything that can be done in exponential time
by an alternating Turing machine with linearly many alternations.
For the two-element BA, it is simply QBF which is maybe the most famous
PSPACE-complete problem. For GS, we saw that the number of formulas
with fixed number of free variables and constants, is triple-exponential
in the number of the free variables. This gives an upper bound for
GS over atomless BA.

\section{Conclusion}

We have presented new methods in the theory of atomless BA that extend
existing results from SBF to BF in general. We also presented the
concept of weakly $\omega$-categorical theories and how they relate
to decidable recurrence relation extensions. We used those ideas to
construct a language that can speak of its own sentences by abstracting
them to merely BA elements. We further extended this construction
to a novel temporal logic with several distinguishing abilities.

\bibliographystyle{plain}
\nocite{*}
\bibliography{time-gs}

\section*{Appendix: Proofs}
\begin{proof}[Proof of \thmref{bcc}]
We prove the first statement and the second is analogous. Further
we prove it only for the univariate case, and the multivariate case
follows immediately by induction. Any univariate BF can be written
in Boole's normal form (sometimes mistakingly called Shannon's normal
form) as $f\left(x\right)=ax\cup bx'$ where $a=f\left(1\right)$
and $b=f\left(0\right)$. Now $f\left(x\right)=0$ iff $ax=bx'=0$,
which reads $b\leq x\leq a'$, so a solution exists iff $b\leq a'$
equivalently $ab=0$.
\end{proof}
\begin{proof}[Proof of \thmref{sol-interval}]
The second equivalence follows immediately from the proof of the
previous theorem. For the first equivalence, write $f\left(x\right)=ax+b$
(this is the algebraic normal form). Then 
\[
f\left(x+f\left(x\right)\right)=a\left(x+ax+b\right)+b=ax+ax+ab+b=ab+b=f\left(0\right)f\left(1\right)=0
\]
and for the other direction, if $f\left(x\right)=0$, just put $t=x$.
\end{proof}
\begin{proof}[Proof of \thmref{main-ineq}]
First assume that $X^{A_{1}},\dots.X^{A_{m}}$ are all distinct and
therefore the nonzero $b$'s are all disjoint, otherwise convert any
two equations of the form $\begin{matrix}X^{A_{i}}\geq s\\
X^{A_{i}}\geq t
\end{matrix}$ into the equivalent form $X^{A_{i}}\geq s\cup t$. Necessity is now
immediate recalling that two different minterms are always disjoint
and that subsets of disjoint sets must also be disjoint. For sufficiency
and $n=1$ the equations take the form $x\geq b_{1}$ and $x'\geq b_{2}$
which indeed has a solution iff $b_{1}b_{2}=0$. Assume for $n$ and
consider a distinguished variable $x$. Then we can split the equations
into $p+q=m$ equations and rewrite them as $\begin{matrix}\bigwedge_{i=1}^{p}xX^{A_{i}}\geq b_{i}\\
\bigwedge_{j=1}^{q}x'X^{B_{j}}\geq c_{j}
\end{matrix}$ and let $X$ be a solution of $\begin{matrix}\bigwedge_{i=1}^{p}X^{A_{i}}\geq b_{i}\\
\bigwedge_{j=1}^{q}X^{B_{j}}\geq c_{j}
\end{matrix}$ using the induction hypothesis after making sure that all $A_{i},B_{i}$
are distinct (while if $p+q=1$ then a solution trivially exists).
If $p\neq0$, set $x=\bigcup_{k}b_{k}$. Then $\bigcup_{k}c_{k}\leq x'$
due to the disjointness assumption. Therefore
\[
xX^{A_{i}}=\left(\bigcup_{k}b_{k}\right)\wedge X^{A_{i}}\geq b_{i}X^{A_{i}}=b_{i}
\]
\[
x'X^{B_{j}}\geq\left(\bigcup_{k}c_{k}\right)X^{B_{j}}\geq c_{j}X^{B_{j}}=c_{j}
\]
Similarly set $x=\bigcap_{k}c_{k}'$ if $p=0$, or simply $x=0$.
\end{proof}
\begin{proof}[Proof of \lemref{qelim-main}]
If $f$ has a zero, then all such zeros are precisely the range of
$x+f\left(x\right)$ by \thmref{sol-interval}. So we can write the
system as $f\left(x+f\left(x\right)\right)=0\wedge\bigwedge_{i}g_{i}\left(x+f\left(x\right)\right)\neq0$.
Now $f$ has a zero iff $f\left(0\right)f\left(1\right)=0$ by Boole's
consistency condition, in which case $f\left(x+f\left(x\right)\right)$
is identically zero.
\end{proof}
\begin{proof}[Proof of \lemref{egsbe-norm}]
First substitute the general solution $x+f\left(x\right)$ of the
positive part into the negative parts and obtain:
\[
f\left(x\right)=0\wedge\bigwedge_{i}g_{i}\left(x+f\left(x\right)\right)\neq0
\]
and since $f\left(x\right)=0$ there is no harm in multiplying the
negative part with $f'\left(x\right)$:
\[
f\left(x\right)=0\wedge\bigwedge_{i}f'\left(x\right)g_{i}\left(x+f\left(x\right)\right)\neq0
\]
now for any $h\left(x\right)$ we have $h\left(x\right)\neq0$ iff
$xh\left(1\right)\neq0\vee x'h\left(0\right)\neq0$ (by Boole's normal
form), so we can rewrite the negative part as:
\[
f\left(x\right)=0\wedge\bigwedge_{i}xf'\left(1\right)g_{i}\left(f'\left(1\right)\right)\neq0\vee x'f'\left(0\right)g_{i}\left(f\left(0\right)\right)\neq0
\]
and using proposition \propref{xfx} for the parts $f'\left(1\right)g_{i}\left(f'\left(1\right)\right)$
and $f'\left(0\right)g_{i}\left(f\left(0\right)\right)$ we obtain
the first result. Now simply account for the conditions of $f,g_{i}$
having zeros at all, and obtain the second result.
\end{proof}
\begin{proof}[Proof of corollary \corref{formula-to-blf}]
Assume $k=1$ for simplicity. Consider $\phi\left(x_{n-1},x_{n},y_{n-1},y_{n}\right)$.
We understand $\phi$ as defining a relation between inputs and outputs
at current time ($x_{n},y_{n}$ respectively) and in the previous
time $x_{n-1},y_{n-1}$. Intuitively, it defines at least one BL$\left[k\right]$
function if the infinitary expression $\forall x_{1}\exists y_{1}\forall x_{2}\exists y_{2}\dots.\bigwedge_{n=2}^{\infty}\phi\left(x_{n-1},x_{n},y_{n-1},y_{n}\right)$
is satisfiable, alternatively if it is true in a model of choice.
This infinitary expression can be given a concrete meaning by considering
the first order theory containing all formulas of the form $\forall x_{1}\exists y_{1}\dots\forall x_{N}\exists y_{N}.\bigwedge_{n=2}^{N}\phi\left(x_{n-1},x_{n},y_{n-1},y_{n}\right)$
for all $N$.
\end{proof}
\begin{proof}[Proof of \thmref{gs-norm}]
It is easy to see that we can always reduce into the collapsed fragment:
this is immediate from the uniqueness of successor, for example $\forall n\exists k.s\left(n,k\right)\wedge\dots$
is same as $\forall nk.s\left(n,k\right)\rightarrow\dots$. For the
normalized form, first convert the formula to DNF at its outermost
level, so each literal may be a complex quantified formula, then collapse
the quantifier alternation as above, so each quantified formula is
either universal or existential. Moving to NNF we can consider universal
and existential literals instead of positive and negative literals.
In each DNF clause we can collapse the universal parts into a single
one since universals distribute over conjunctions. Given an existential
literal $\exists T.\gamma\left(T\right)\wedge\phi$ while denoting
$T=t_{1},\dots,t_{k}$, we introduce a flag $e$ which is an additional
output variable, and write
\[
\left[\exists t.e\left(t\right)=0\right]\wedge\forall T.\gamma\left(T\right)\rightarrow\left[e\left(k-1\right)=1\wedge\left(e\left(t_{k}\right)=0\leftrightarrow\left(\psi\vee e\left(t_{k}-1\right)=0\right)\right)\right]
\]
where $t_{k}=\max\left\{ t_{1},\dots,t_{k}\right\} $ is assumed.
The existential part is therefore reduced into a single atom at the
expense of introducing a new output stream, and with introducing new
universal literals which can then be collapsed into a single one as
above. Given multiple single-atom existential parts $\bigwedge_{k}\exists t.e_{k}\left(t\right)=0$
we can easily see that they are equivalent to $\exists t.\left[\bigcup_{k}e_{k}\left(t\right)\right]=0$
because each flag remains zero once it becomes zero, so there is a
point in time where all flags are eventually zero, so the existential
part can be merely a single $\exists n.e\left(n\right)=0$ by defining
this additional flag in the universal part. By that we reduced both
the universal and the existential parts into a single one each.
\end{proof}
\begin{proof}[Proof of \thmref{gs-main}]
A model of size $n+1$ exists iff $\forall x_{0}\exists y_{0}\dots\forall x_{n}\exists y_{n}.\bigwedge_{m=k}^{n}\phi\left(X_{m}^{i},Y_{m}^{i}\right)$.
Leaving free the first $k+1$ time points we can write
\[
\phi_{n-k}\left(X_{k}^{i},Y_{k}^{i}\right):=\forall x_{k+1}\exists y_{k+1}\dots\forall x_{n}\exists y_{n}.\bigwedge_{m=k}^{n}\phi\left(X_{m}^{i},Y_{m}^{i}\right)
\]
\[
=\forall x_{k+1}\exists y_{k+1}\dots\forall x_{n}\exists y_{n}.\phi\left(X_{k}^{i},Y_{k}^{i}\right)\wedge\bigwedge_{m=k+1}^{n}\phi\left(X_{m}^{i},Y_{m}^{i}\right)
\]
\[
=\phi\left(X_{k}^{i},Y_{k}^{i}\right)\wedge\forall x_{k+1}\exists y_{k+1}\dots\forall x_{n}\exists y_{n}.\bigwedge_{m=k+1}^{n}\phi\left(X_{m}^{i},Y_{m}^{i}\right)
\]
\[
=\phi\left(X_{k}^{i},Y_{k}^{i}\right)\wedge\forall x_{k+1}\exists y_{k+1}.\phi_{n-k-1}\left(X_{k+1}^{i},Y_{k+1}^{i}\right)
\]
since replacing $k$ with $k+1$ in $\phi_{n-k}\left(X_{k}^{i},Y_{k}^{i}\right):=\forall x_{k+1}\exists y_{k+1}\dots\forall x_{n}\exists y_{n}.\bigwedge_{m=k}^{n}\phi\left(X_{m}^{i},Y_{m}^{i}\right)$
results with $\phi_{n-k-1}\left(X_{k+1}^{i},Y_{k+1}^{i}\right):=\forall x_{k+2}\exists y_{k+2}\dots\forall x_{n}\exists y_{n}.\bigwedge_{m=k+1}^{n}\phi\left(X_{m}^{i},Y_{m}^{i}\right)$.
In case that $\forall x_{0}\exists y_{0}\dots\forall x_{k}\exists y_{k}.\phi_{\infty}\left(X_{k}^{i},Y_{k}^{i}\right)$
then due to monotonicity, every $k+1$ subsequence of time points
will have to satisfy $\phi_{\infty}\left(X_{k}^{i},Y_{k}^{i}\right)$,
and any such subsequence can be extended arbitrarily due to the fact
that it is a fixed point indeed.
\end{proof}

\end{document}